\RequirePackage[l2tabu, orthodox]{nag}
\documentclass[10pt, twocolumn, a4paper]{IEEEtran}

\usepackage[T1]{fontenc}
\usepackage{amssymb}
\usepackage{amsfonts}
\usepackage{epsfig}
\usepackage[caption=false]{subfig}
\usepackage{calc}
\usepackage{cite}
\usepackage{color}
\usepackage{bm}
\usepackage{enumerate}
\usepackage{framed}
\usepackage{arydshln}
\usepackage[fleqn]{amsmath}
\usepackage{amsthm}
\usepackage{nopageno}
\usepackage{url}
\usepackage[nosort]{cleveref}

\newtheorem{thm}{Theorem}[section]
\newtheorem{cor}[thm]{Corollary}

\newtheorem{lem}[thm]{Lemma}

\newcommand{\bit}{\begin{itemize}}
\newcommand{\eit}{\end{itemize}}
\newcommand{\bcor}{\begin{cor}}
\newcommand{\ecor}{\end{cor}}

\newcommand{\fset}{\mathcal{R}}

\title{The Storage vs Repair Bandwidth Trade-off for Multiple Failures in Clustered Storage Networks}
\author{\IEEEauthorblockN{Vitaly Abdrashitov, N. Prakash and Muriel M\'edard}\\
	\IEEEauthorblockA{Research Laboratory of Electronics, MIT, USA. \ \ Email: \{vit, prakashn, medard\}@mit.edu.}
	\thanks{This work is in part supported by the Air Force Office of Scientific Research (AFOSR) under award No FA9550-14-1-043, and in part supported by the National Science Foundation (NSF) under Grant No.  CCF-1527270.}
}

\begin{document}

\maketitle
\thispagestyle{empty}
\pagestyle{empty}

\begin{abstract}
We study the trade-off between storage overhead and inter-cluster repair bandwidth in clustered storage systems, while recovering from multiple node failures within a cluster. A cluster is a collection of $m$ nodes, and there are $n$ clusters. For data collection, we download the entire content from any $k$ clusters. For repair of $t \geq 2$ nodes within a cluster, we take help from $\ell$ local nodes, as well as $d$ helper clusters.  We characterize the optimal trade-off under functional repair, and also under exact repair for the minimum storage  and minimum inter-cluster bandwidth (MBR) operating points. Our bounds show the following interesting facts: $1)$ When $t|(m-\ell)$ the trade-off is the same as that under $t=1$, and thus there is no advantage in jointly repairing multiple nodes, $2)$ When $t \nmid (m-\ell)$, the optimal file-size at the MBR point under exact repair can be strictly less than that under functional repair. $3)$ Unlike the case of $t=1$, increasing the number of local helper nodes does not necessarily increase the system capacity under functional repair.
\end{abstract}

\vspace{-0.1in}

\section{Introduction}

We study the storage-overhead vs repair-bandwidth trade-off for multiple node failures under the setting of  clustered storage networks. Our model is motivated by applications to cloud storage settings, where user data is spread across distinct data-centers, even possibly belonging to different service providers (as in a cloud-of-cloud setting). Practical implementation studies that show the benefits of Reed-Solomon codes for data storage in cloud-of-cloud settings appear in \cite{racs, depsky, cyrus}. In our model, a cluster represents a data center. In such networks, it is common to differentiate between intra- and inter-cluster bandwidth costs; typically, intra-cluster bandwidth cost is much less than inter-cluster bandwidth cost. To keep the model simple, we ignore any hierarchical topology that may be present within a data-center (cluster), and simply assume equal cost connectivity between any two nodes inside a cluster. We also assume direct connectivity between any two clusters in the network.

In our model, a cluster is a collection of $m$ physical nodes (see Fig. \ref{fig:sys_model}), each of size $\alpha$ symbols from the finite field $\mathbb{F}_q$, for some $q$. There are $n$ clusters in total in the system. A file of size $B$ symbols is encoded into $nm\alpha$ symbols and stored across the $nm$ storage nodes. We follow a clustering approach for both data collection and repairs. For data collection, we demand that the entire content of an arbitrary set of $k$ clusters is sufficient to decode the original file. Thus, during data collection, we assume a cluster to be either completely available (if we connect to it) or completely unavailable (if we do not connect to it). Such an assumption is realistic in a multi-data-center cloud setting~\cite{hitachi}.

\begin{figure}
	\centering
	\includegraphics[width=65mm]{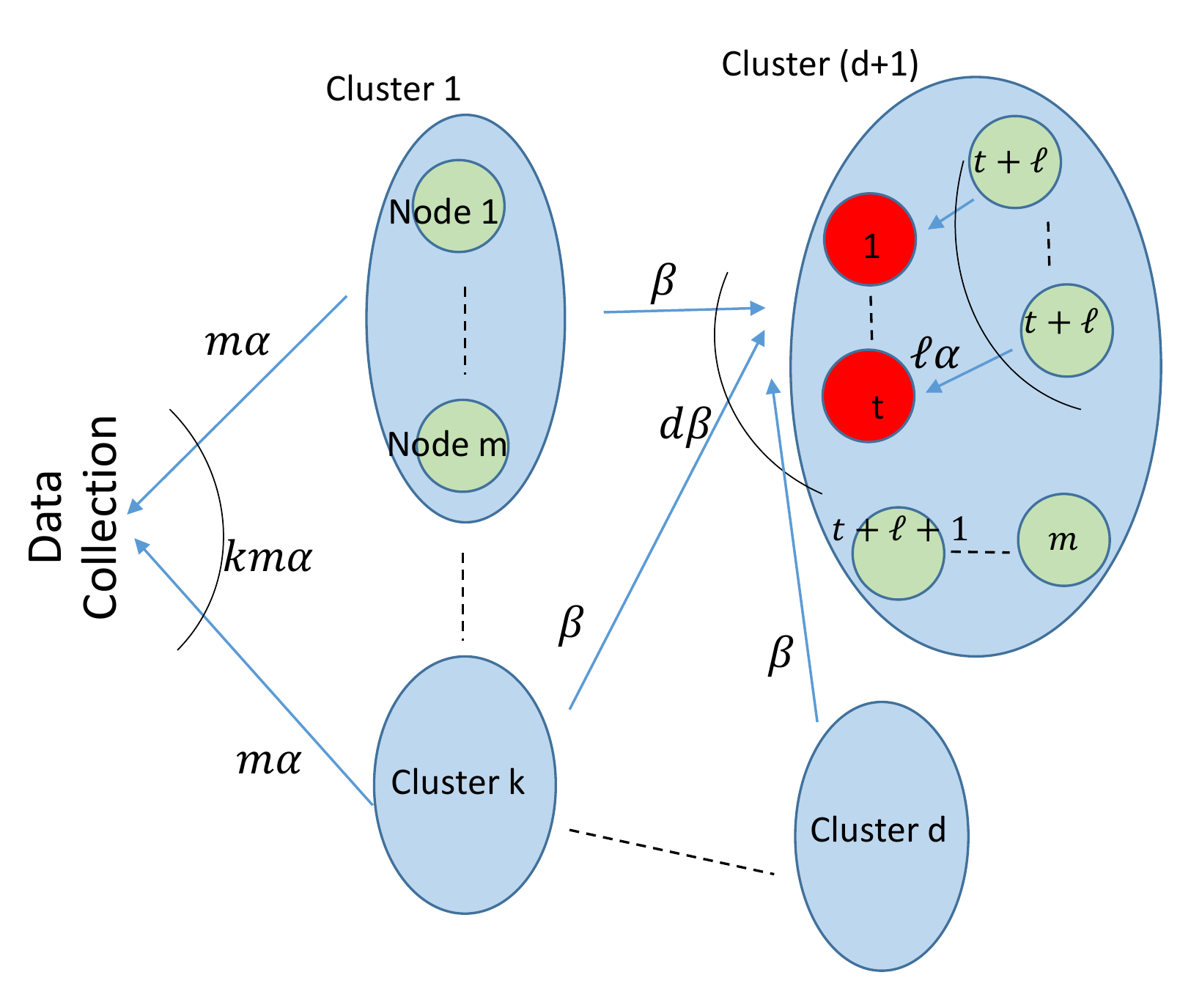}
	\caption{System model for clustered data storage, illustrating data collection and repair from multiple failures. The model is suitable for storing coded data across multiple data centers (clusters) as in cloud-of-cloud applications.}
	\label{fig:sys_model}
\end{figure}
\begin{table}[tb]
	\centering
	\begin{tabular}{ |c|c| }
		\hline
		Special Case of Our Model& Prior Work \\
		\hline \hline
		No clustering, single-node  & Classical RC, \cite{dimakis}  \\
		repair: $m = 1, t = 1, \ell = 0$ & \\
		\hline
		With clustering, single-node &  Generalized RC, \cite{genrcTIT},  \\
		repair: $t = 1$ & Studies impact of $\ell$ on \\
		& $1)$ storage vs  Inter-cluster BW trade-off \\
		& $2)$ Intra-cluster BW \\
		\hline
		With clustering, multiple-node&  Two Layer coding scheme~\cite{skoglund_partial}\\
		repairs: $t \geq 1$ &  Study limited to $\ell = m - t$ \\
		\hline
	\end{tabular}
	\caption{Special Cases of System Model Appearing in Literature.}
	\label{tab:special_cases}
	\vspace{-0.3in}
\end{table}

\begin{figure*}
	\centering
	\subfloat[Trade-off for an $(n = 5, k = 4, d = 4)(m =3, \ell = 0, t = 2)$ system.]{\label{fig:tradeoff}\includegraphics[height=1.5in]{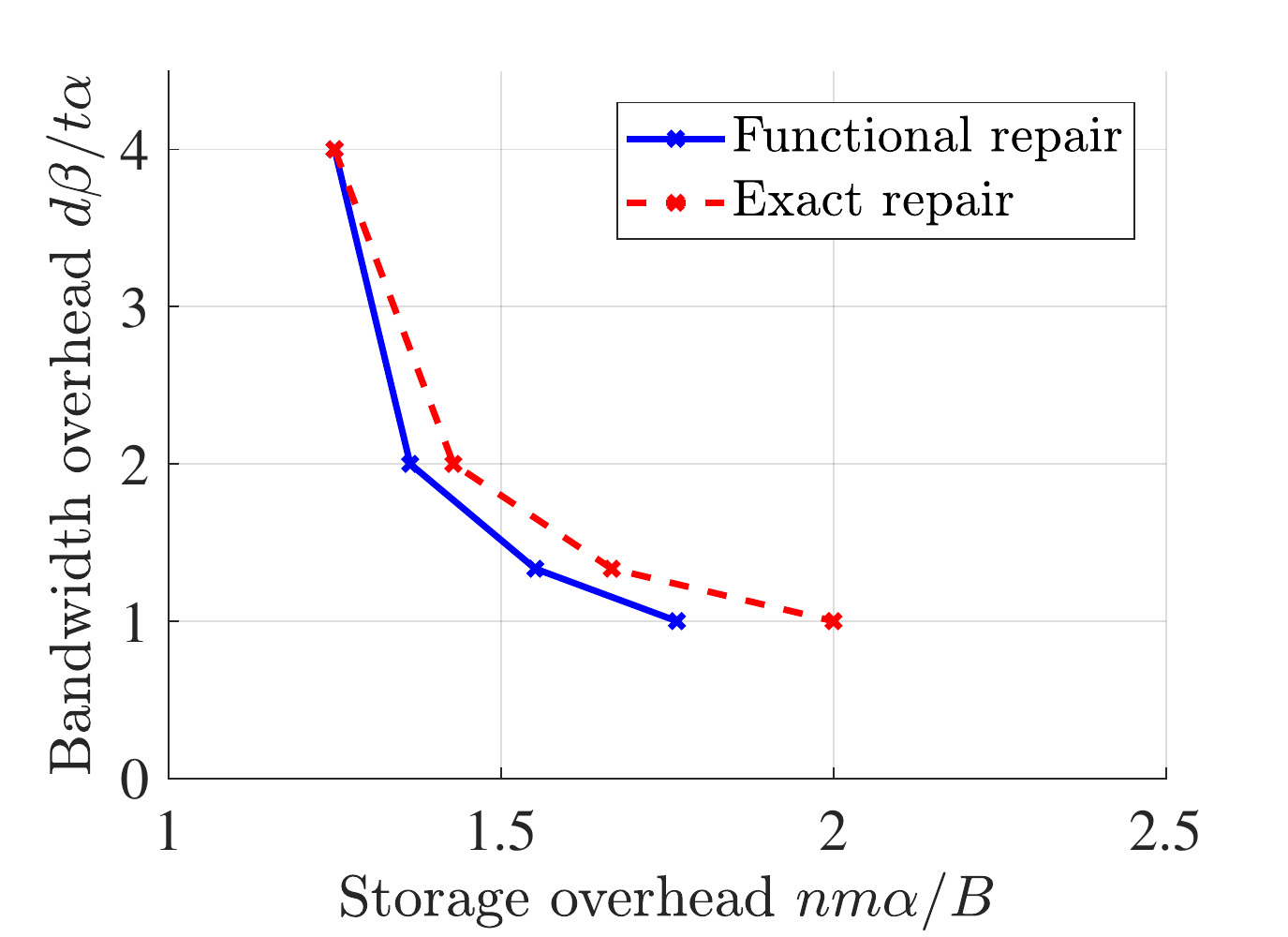}}
	\hfil
	\subfloat[RLNC simulation for an $(n  = 3, k = 2, d = 2)(\alpha = 2, \beta = 2)(m = 3, \ell = 0, t=2)$ system.]{\label{fig:rlnc_filesize}\includegraphics[height=1.5in]{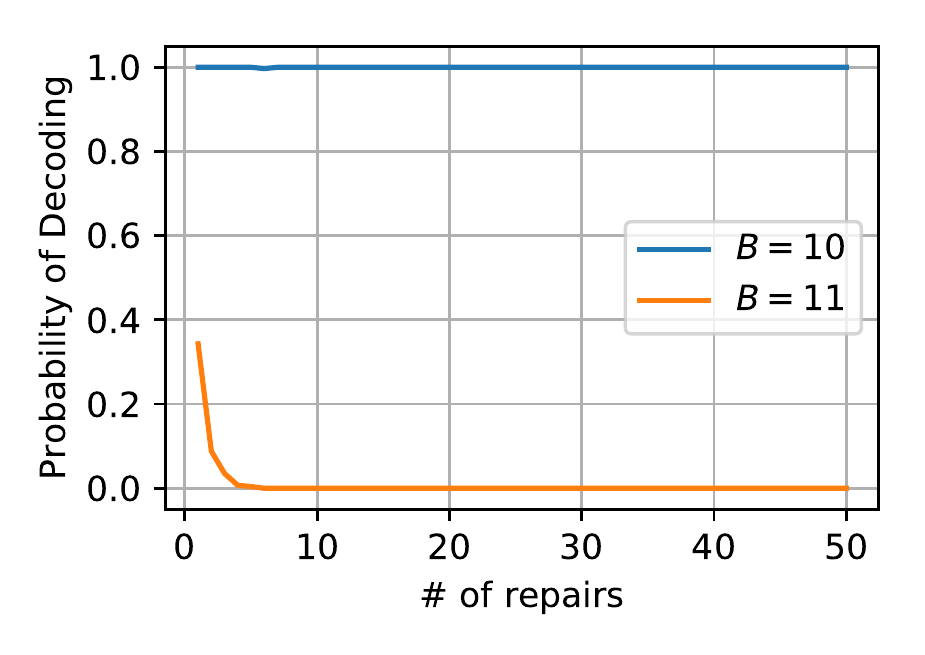}}
	\hfil
	\subfloat[Impact of number of local helper nodes, $\ell$, on file-size for an $(n = 7, k = 4, d = 5, m=17, t=5)$ clustered storage system at MBR point $(\alpha=1, \beta=1)$. Local help does not provide any advantage unless $\ell > 2$.]{\label{fig:ell_plot}\includegraphics[height=1.5in]{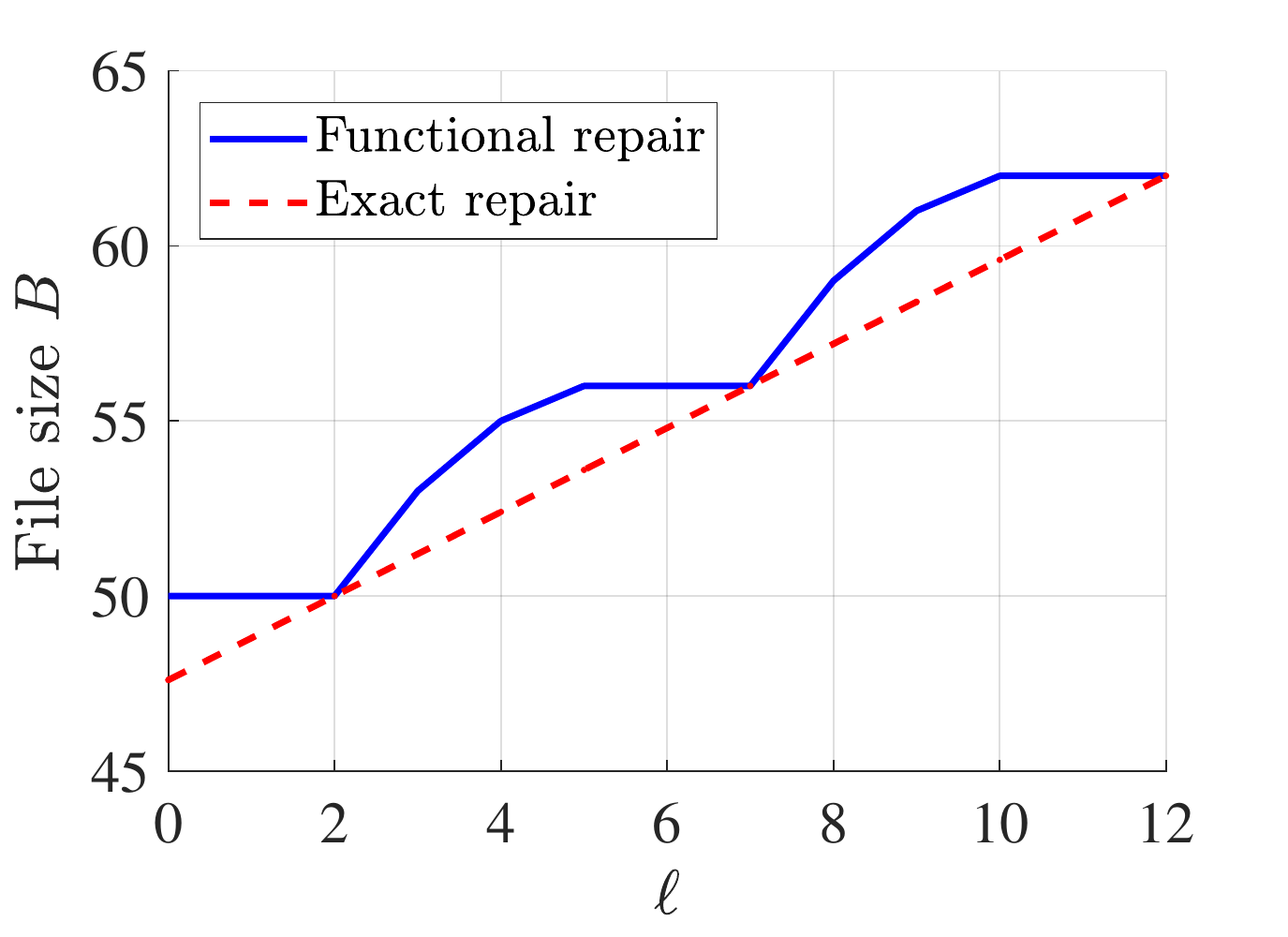}}
	\caption{Illustration of the implications of the exact and functional repair file-size bounds.}
	\label{fig:implication}
\end{figure*}

Nodes within a cluster represent failure domains; in this work, we deal with the problem of recovery from  $t$ node-failures that occur in one of the $n$ clusters. While single-node failure is the most common failure event, correlated failures of nodes within a data center is an important issue reported in practice \cite{ford2010availability} and this motivates our failure model. The $t$ \emph{newcomer} nodes are added to the same cluster as replacement to those failed. For restoring the content of the $t$ new nodes, we download local as well as external content. External help is taken from  any set of $d$ other clusters, each of which contributes $\beta$ inter-cluster symbols. The $\beta$ symbols from any cluster is possibly a function of all the $m\alpha$ symbols in the cluster. For completeness of the model, we assume the presence of a compute unit in the cluster that can combine these $m\alpha$ symbols to generate the $\beta$ helper symbols. We also download entire content from any set of $\ell \leq m - t$ surviving nodes in the failure cluster. Once again, we assume the presence of a compute unit in the failure cluster that combines all the local and external helper data, and generates the content of the replacement nodes. We assume that the encoding function does not introduce any local dependence among the nodes of a cluster; for e.g., the model excludes the possibility of a local parity node within a cluster. An analysis of the impact of such local parity nodes  is left for future work. We also restrict ourselves to the case $d \geq k$, even though analysis for the case $0 \leq d \leq k-1$ is perfectly feasible.

A code satisfying the above model requirements for repair and data collection shall be called multi-node repair generalized regenerating code (MRGRC) $\mathcal{C}$  with parameters $\{(n,k,d),(\alpha, \beta),(m,\ell,t)\}$. In this paper, we study the trade-off between storage-overhead (S.O.) $\frac{mn\alpha}{B}$  and inter-cluster (IC) repair-bandwidth-overhead $\frac{d\beta}{t\alpha}$ for the above setting, under both functional and exact repair. Under exact repair, the content of any of the $t$ new nodes is exactly the same as what was stored before failure, while in functional repair, the restored content allows data-collection and further repairs.

Special cases of the model have been studied in the past (see Table \ref{tab:special_cases}). The setting of  regenerating codes (RC) introduced in \cite{dimakis} corresponds to the case with $t = 1, m = 1, \ell = 0$ - we refer to these as classical RCs. The case of single node failure ($t = 1$) in clustered systems was previously studied in \cite{genrcTIT}, where the authors first identify the storage-vs-inter-cluster-BW trade-off (ignoring intra-cluster BW), and then find bounds on the minimum intra-cluster BW that is needed to achieve this trade-off. The authors, show the surprising fact that while increasing the number of local helper nodes $\ell$ improves the storage-vs-inter-cluster-repair-BW trade-off; increasing $\ell$ not only increases intra-cluster BW in the host-cluster (this is obvious since one downloads $\ell \alpha$ symbols), but also increases the intra-cluster BW in the $d$ remote helper clusters. In other words, in situations when intra-cluster BW cannot be entirely ignored, the choice of the number of local nodes becomes an important one.

Motivated by the above result of \cite{genrcTIT}, for the case of multiple failures that we consider here, even though we do not explore bounds on intra-cluster BW in this paper, we still parametrize the number of local helper nodes $\ell$ in the range $0 \leq \ell \leq m-t$, so that our results remain relevant for a future study on intra-cluster BW for $t \geq 1$. The case of $\ell = m - t, t \geq 1$ has been previously studied in \cite{skoglund_partial}. However, as we show in this paper, even when restricted to storage-vs-inter-cluster-BW trade-off (ignoring intra-cluster BW), the case $0 \leq \ell < m-t, t > 1$ offers several surprising results which cannot be inferred from an analysis of the case $\ell = m-t, t > 1$. Following is a summary of our results in this paper:
\vspace{-.1in}
\subsection{Our Results}
\vspace{0.05in}
\emph{(a)}\emph{\underline{File-Size bound under functional repair:}}  Let $(m-\ell) = at + b, a \geq  1, 0 \leq b \leq t-1$. Then, the file-size under functional repair of an MRGRC is upper bounded by $B \leq  B_F^*$, where

{\small
\begin{equation}
 B_F^*  =  \ell k\alpha + a \sum_{i=0}^{k-1} \min (t\alpha, (d-i)\beta ) +  \sum_{i=0}^{k-1} \min (b\alpha, (d-i)\beta ). \label{eq:func_file_size}
\end{equation}
}

The bound is shown by considering the information-flow graph (IFG) under functional-repair, and calculating the minimum cut. The bound is indeed tight, if there is a known upper bound on the number of repairs in the system - the achievability follows from results in network coding~\cite{KoetterMedard}.

\vspace{0.05in}
\emph{(b)}\emph{\underline{File-Size bound under exact repair:}} For exact repair, we prove a tighter bound, given by
{\small
\begin{eqnarray}
 B & \leq & B_E^* = \ell k \alpha + (m-\ell) \sum_{i = 0}^{k-1}\min\left(\alpha, \frac{(d-i)\beta}{t}\right). \label{eq:exact_file_size}
\end{eqnarray}}
We note that $B_E^* \leq B_F^*$. The bound is optimal at the minimum storage-overhead (MSR) and the minimum inter-cluster repair-bandwidth-overhead (MBR) points characterized by $B = m k \alpha$ and $t\alpha = d \beta$, respectively. We show how optimal constructions for the case $t >1$ can be directly obtained from optimal constructions for the case $t = 1$~\cite{genrcTIT}.

\vspace{0.05in}
\emph{\underline{Implications of the Bounds:}} \emph{Case a) $t|(m-\ell)$:} In this case, the bounds in \eqref{eq:func_file_size} and \eqref{eq:exact_file_size} coincide.
Further, \eqref{eq:func_file_size} gives the same S.O. vs IC-repair-bandwidth-overhead trade-off for any value of $t \geq 1$. i.e.,  under functional repair, there is no advantage to jointly repairing multiple nodes (instead of repairing one ). For exact repair, at the MSR and MBR points, there is no benefit to jointly repairing multiple nodes for any $t > 1$, irrespective of if $t|(m-\ell)$ or not.

\emph{Case b) $t\nmid(m-\ell)$:} In this case, it is possible that $B_F^* > B_E^*$. Specifically, at the MBR point with $t \alpha = d \beta$, we have $B_F^* > B_E^*$, whenever $k > 1$. This also means that the S.O. vs IC-repair-bandwidth-overhead trade-off under functional repair for the case $t >1$ (with $k > 1$) is strictly better than that for the case $t = 1$.  A comparison of trade-offs between exact and functional repair for the case of $\{(n = 5, k = 4, d = 4)(m = 3, \ell = 0, t = 2)\}$ is shown in Fig. \ref{fig:tradeoff}. In Fig. \ref{fig:rlnc_filesize}, we present a  simulation result that shows probability of successful decoding while using random linear network codes (RLNCs)~\cite{rlnc} with sufficiently large field size in an $\{(n = 3, k = 2, d = 2)(m = 3, \ell = 0, t = 2)\}$ storage system operating at the MBR point with $\beta = 2$. In this case, optimal file-sizes under exact and functional repair are $B_E^* = 9$ and $B_F^* = 10$. RLNCs enable functional repair, and our simulation result indeed confirms the achievability of file-size $B_F^* = 10$.

Another implication of the bounds relates to the usefulness of the number of local helper nodes $\ell$ used in the repair process. Under functional repair, for the case of $t =1$~\cite{genrcTIT}, if we fix $n, k, d, m, t, \alpha, \beta$, the optimal file-size increases strictly monotonically with $\ell$, whenever $\alpha > (d-k+1)\beta$ (i.e., if we exclude the MSR point) . However, strict monotonicity is not necessarily true when $t > 1$.  Specifically, at the MBR point, it can be shown that whenever $(m \mod t) \leq \lfloor (d-k+1)t/d \rfloor$, for any $\ell$ in the range $0 \leq \ell \leq (m \mod t)$,   the capacity is as good as with no local help at all (see Fig. \ref{fig:ell_plot}).

\vspace{-0.1in}
\subsection{Other Related Work}

The problem of multiple-node repair for classical RCs have been studied under the frameworks of cooperative repair~\cite{shum_coop, kermarrec_coop} and centralized repair~\cite{cadambe_asymptotic,ankit_centralized}.
In cooperative repair, each of the $t$ replacement nodes first individually contacts respective sets of $d$ helper nodes, and then communicates among themselves before restoring the new content. In centralized repair, a centralized compute node downloads data from some subset of $d$ nodes, and generates the data for all $t$ replacement nodes. Our repair model can be considered as a centralized repair model for clustered storage systems.

Regenerating code variations for data-center-like topologies consisting of racks and nodes are considered in \cite{plee_isit2016_doubleregen, clust_stor_Moon, gaston2013realistic, Gaston_nonhom, ozan_xyregen}. All these works focus on single-node repair, whereas we focus on multiple-node repairs. Further, the models in \cite{plee_isit2016_doubleregen}, \cite{clust_stor_Moon} and \cite{gaston2013realistic} use  clustering approach only for repair (by distinguishing inter and intra rack repair costs), and not for data-collection.
File retrievability is demanded from any set of $k$ nodes in the whole system, irrespective of which clusters they belong to. The difference in data collection model is the main difference between our model and the models in \cite{plee_isit2016_doubleregen}, \cite{clust_stor_Moon} and \cite{gaston2013realistic}.

We next describe how exact repair codes for $t > 1$ can be directly obtained from exact repair codes for $t=1$. In Sections \ref{sec:exact} and \ref{sec:func} we discuss the exact-repair and functional-repair bounds, respectively. For functional repair, our IFG model is substantially different, and more elaborate than the one used in \cite{skoglund_partial} for the case of $\ell = m-t$. The  complexity of our model comes from the need to handle the case $\ell < m-t$.

 \section{Exact Repair Codes}
 Optimal constructions of exact repair MRGRCs for any $t > 1$ can be directly obtained from constructions for the case $t = 1$, whenever $t|\beta$. In order to construct an exact repair MRGRC $\mathcal{C}$ with parameters $(n, k, d)(\alpha, \beta)(m, \ell, t), t|\beta$, we start with an exact repair code~\cite{genrcTIT} $\mathcal{C}'$ with parameters $\{(n, k, d)(\alpha, \beta' = \beta/t)(m, \ell, t' = 1)\}$. The code $\mathcal{C}'$ was shown to exist at the MSR and MBR points; in fact it was shown in \cite{genrcTIT} that an optimal $(n, k, d)(\alpha, \beta')(m, \ell, t = 1)$ $\mathcal{C}'$ can be constructed whenever a classical exact repair  ${(n, k, d)(\alpha, \beta')}$ RC exists, with file-size $\sum_{i=0}^{k-1}\min(\alpha, (d-i)\beta')$.

 The code  $\mathcal{C}'$ can be viewed as the code $\mathcal{C}$ as it is, if we assume that  repair of any group of $t$ nodes in $\mathcal{C}$ happens one node at a time via the repair procedure in $\mathcal{C}'$.  Also, we use the same set of local and external helpers for the repair of any of the $t$ failed nodes. Inter-cluster bandwidth, for the repair of the entire group, per external helper amounts to $\beta = t\beta'$. The file-size $B$ that we obtain is given by
 \begin{IEEEeqnarray}{rCl}
     B =  B' & = & \ell k \alpha + (m-\ell)\sum_{i = 0}^{k-1}\min(\alpha, (d-i)\beta') \nonumber \\
     & = & \ell k \alpha + (m-\ell)\sum_{i = 0}^{k-1}\min\left(\alpha, \frac{(d-i)\beta}{t}\right).
 \end{IEEEeqnarray}
 \vspace{-0.1in}
 \section{File Size bound, exact repair} \label{sec:exact}
 In this section, we present the proof of the file-size upper bound in \eqref{eq:exact_file_size} for exact repair codes. We assume the code to be deterministic; by this we mean that the helper data is uniquely determined given the indices of the $t$ failed nodes, local helper nodes and helper clusters. We begin with useful notation.
 Let $\mathcal{F}$ denote the random variable corresponding to the data file that gets stored. We assume $\mathcal{F}$ to be uniformly distributed over  $\mathbb{F}_q^{B}$. Let $Y_{i, j} \in \mathbb{F}_q^{\alpha}, 1 \leq i \leq n, 1 \leq j \leq m$ denote the content stored in node $j$ of cluster $i$. For $j \leq j'$, we write $Y_{i,[j , j']}$ to denote $\{Y_{i, j}, Y_{i, j+1}, \ldots, Y_{i, j'}\}$.
 We also write ${\bf Y}_i$ to denote $Y_{i, [1 , m]}$. Further, for $i \leq i'$, ${\bf Y}_{[i , i']}$ will denote $\{ {\bf Y}_i, \ldots, {\bf Y}_{i'} \}$. The property of data collection demands that
 \begin{eqnarray} \label{eq:data_collect}
 H\left(\mathcal{F} | \{{ \bf Y}_i, i \in S \}\right) & = & 0 \ \forall S \subset [n], |S| = k,
 \end{eqnarray}
 where $H(.)$ denotes the entropy function computed with respect to $\log q$. Next, consider the repair of $t$ nodes indexed by $\mathcal{R}_i$ in cluster $i$. Let $\mathcal{H} \subset [n]\backslash\{i\}, |\mathcal{H}| = d$, and $\mathcal{L} \subset [m]\backslash \mathcal{R}_i, |\mathcal{L}| = \ell$ respectively denote the indices of helper clusters and local nodes that aid in the repair process. Let $Z_{i', \mathcal{R}_i}^{\mathcal{H}, \mathcal{L}}$ denote helper data passed by cluster $i'$. The property of exact repair is jointly characterized by the following set of inequalities: $H\left(Z_{i', \mathcal{R}_i}^{\mathcal{H}, \mathcal{L}} | {\bf Y}_{i'} \right)  =  0$, $H\left(Z_{i', \mathcal{R}_i}^{\mathcal{H}, \mathcal{L}}\right)  \leq  \beta$, and
 {
 	\begin{eqnarray}
 	H\left(\{Y_{i, j}, j \in \mathcal{R}_i\}| \{Z_{i', \mathcal{R}_i}^{\mathcal{H}, \mathcal{L}}, Y_{i, j'}, i' \in \mathcal{H},  j' \in \mathcal{L}\}\right)& = & 0, \nonumber \\
 	& & \hspace{-2.75in}   \forall \mathcal{H}  \subset [n]\backslash \{i\}, |\mathcal{H}| = d, \forall \mathcal{L} \subset [m]\backslash \mathcal{R}_i, |\mathcal{H}| = \ell. \label{eq:exact_rep}
 	\end{eqnarray}
 }
 Our proof technique of the file-size bound presented here, though has some similarity with the information theoretic techniques in works like \cite{ankit_centralized}, \cite{rbt}, it differs in an important way.  The proofs in these other works rely on the chain rule of entropy, and so does our proof; however in here we demand that the chain be expanded in a specific order. The following lemma is used to determine this order. The lemma is valid only when $b > 0$, where $(m-\ell) = at + b, a \geq 1, 0 \leq b \leq t-1$. When $b = 0$, the bound-proof does not need this lemma.

 \begin{lem} \label{lem:exact}
 	Let $(m-\ell) = at + b, a \geq 1, 1 \leq b \leq t-1$. Consider any $S_i \subset [n], |S_i| = i, 1 \leq i \leq k-1$, and let $\mathcal{Y}(S_i) = \{ {\bf Y}_{i}, i \in S_i \}$.  Then, for any $i' \in [n]\backslash S_i$, there exists a permutation $\sigma_{i', S_i}$ of  $\{\ell+1, \ell + 2, \ldots, m\}$
 	such that
 	{
 		\begin{equation} \label{eq:lem_exact1}
 		H\left(Y_{i', \sigma_{i', S_i}(j')}|\mathcal{Y}(S_i), \widetilde{\mathcal{Y}}(i', S_i, j') \right)  \leq  \min\left(\alpha, \frac{(d-i)\beta}{t}\right),
 		\end{equation}
 	}
 	for all $j' \in \{m - b + 1, m - b + 2, \ldots, m\}$, where
 	{
 		\begin{equation} \label{eq:lem_exact2}
 		\widetilde{\mathcal{Y}}(i', S_i,j')  = Y_{i', [1 , \ell]} \cup  \{Y_{i', \sigma_{i', S_i}(j)}, j \in [\ell+1 , j'-1 ] \}.
 		\end{equation}
 	}
 \end{lem}
 \begin{proof}
 	In here, we only present the candidate for the permutation $\sigma_{i', S_i}$.  The proof that this satisfies the lemma can be found in Appendix A.  Consider the content, of the cluster $i'$, given by  $\{Y_{i', 1}, Y_{i', 2}, \ldots, Y_{i',m} \}$. Define the quantities $(j_{m}, \mathcal{V}_{m}), (j_{m-1}, \mathcal{V}_{m-1}), \ldots, (j_{m - b+1}, \mathcal{V}_{m - b + 1})$ in this respective order as below:
 	\begin{framed}
 		\begin{enumerate}[Step 1.]
 			\item Let $ \mathcal{U} = \{Y_{i', \ell + 1}, Y_{i', \ell+2}, \ldots, Y_{i',m} \}$, and $x = 0$
 			\item  Define $(j_{m - x}, \mathcal{V}_{m - x})$ as
 			\begin{eqnarray} \label{eq:def_j}
 			(j_{m - x}, \mathcal{V}_{m - x}) & = & \arg \min_{\substack{(j, \mathcal{V}) \ : \ \\ Y_{i', j} \in \mathcal{U} \\ \mathcal{V} \subset \mathcal{U} \backslash \{Y_{i', j}\}, |\mathcal{V}| = t-1}} \Theta, \nonumber
 			\end{eqnarray}
 			where $\Theta = H\left(Y_{i',j}|\mathcal{V}, \mathcal{Y}(S_i), Y_{i', [1:\ell]} \right)$.
 			\item If $x < b-1$, update $\mathcal{U}$ as $\mathcal{U} = \mathcal{U} \backslash \{ Y_{i', j_{m - x}}\}$. Increment $x$ by $1$ and return to Step $2$.
 		\end{enumerate}
 	\end{framed}
 	Additionally, let us also define
 	$\{ j_{\ell + 1}, j_{\ell + 2}, \ldots, j_{m - b} \}  \triangleq  \{\ell + 1, \ldots, m\}$ $\backslash \ \{j_{m}, j_{m-1}, \ldots, j_{m-b+1} \}$. In the preceding  definition, we only need equality as sets. We do not care about any particular ordering of the elements in  $\{\ell + 1, \ldots, m\} \ \backslash \ \{j_{m}, j_{m-1}, \ldots, j_{m-b+1} \}$, while associating these with $\{ j_{\ell + 1}, j_{\ell + 2}, \ldots, j_{m - b} \}$. The candidate for the permutation $\sigma_{i', S_i}$ on the set $\{\ell + 1, \ldots, m\}$ is now defined as follows:
 	\begin{eqnarray} \label{eq:perm}
 	\sigma_{i', S_i}(p) = j_{p}, \ \ell + 1 \leq p \leq m.
 	\end{eqnarray}
 \end{proof}

 \emph{\underline{Proof of Exact Repair Upper Bound}:} We have
 {\small
 	\begin{eqnarray}
 	B & = & H(\mathcal{F})  \leq  H({\bf Y}_{[1  , k]})   =  \sum_{i' = 1}^{k} H({\bf Y}_{i'} | {\bf Y}_{[1  , i'-1]}) \nonumber \\
 	&  \hspace{-0.5in} =  & \hspace{-0.4in} \sum_{i' = 1}^{k} \left( H(Y_{i', [1  , \ell]}|{\bf Y}_{[1  , i'-1]}) +
 	H(Y_{i', [\ell + 1  , m]}|Y_{i', [1  , \ell]}, {\bf Y}_{[1  , i'-1]}) \right)
 	\nonumber \\
 	& \leq & \ell k \alpha +  \sum_{i' = 1}^{k}
 	H(Y_{i', [\ell + 1  , m]}|Y_{i', [1  , \ell]}, {\bf Y}_{[1  , i'-1]}). \label{eq:lemma_exact_1}
 	\end{eqnarray}
 }
 Now, if we let $\sigma = \sigma_{i', [1 , i'-1]}$ to be the permutation as obtained from Lemma \ref{lem:exact}, then we expand the term $H(Y_{i', [\ell + 1 , m]}|Y_{i', [1 , \ell]}, {\bf Y}_{[1 , i'-1]})$ in \eqref{eq:lemma_exact_1} using the order determined by the permutation $\sigma$, as follows:
 \begin{IEEEeqnarray}{rCl}
 	\IEEEeqnarraymulticol{3}{l}{H(Y_{i', [\ell + 1 , m]}|Y_{i', [1 , \ell]}, {\bf Y}_{[1 , i'-1]})}  \nonumber \\
 	&  = &  H(\{Y_{i', \sigma(j')}, j' \in [\ell + 1 , m]\}|Y_{i', [1 , \ell]}, {\bf Y}_{[1 , i'-1]})   \nonumber \\
 	&  \leq & \sum_{u = 0}^{a-1}H(\{Y_{i', \sigma(\ell+ut+v)}, v \in [1 , t]\}|Y_{i', [1 , \ell]}, {\bf Y}_{[1 , i'-1]}) \nonumber \\
 	& & + \sum_{j' = m - b + 1}^{m}H(Y_{i', \sigma(j')}, |{\bf Y}_{[1 , i'-1]}, \widetilde{\mathcal{Y}}(i', [1 , i'-1],j')), \IEEEeqnarraynumspace \label{eq:lemma_exact_2}
 \end{IEEEeqnarray}
 where $\widetilde{\mathcal{Y}}(i', [1 , i'-1], j')$ is defined using \eqref{eq:lem_exact2}. Using \eqref{eq:exact_rep}, each term under the first summation in \eqref{eq:lemma_exact_2} is upper bounded by $\min(t\alpha, (d-i'+1)\beta)$, while each term under the second summation in  \eqref{eq:lemma_exact_2} is upper bounded using Lemma \ref{lem:exact}. Thus, we get that
 \begin{IEEEeqnarray}{rCl}
 	\IEEEeqnarraymulticol{3}{l}{H(Y_{i', [\ell + 1 , m]}|Y_{i', [1 , \ell]}, {\bf Y}_{[1 , i'-1]})} \nonumber \\
 	& \leq &  a \min(t\alpha, (d-i'+1)\beta) +b \min\left(\alpha, \frac{(d-i'+1)\beta}{t}\right) \nonumber \\
 	&  =  &  (m-\ell) \min\left(\alpha, \frac{(d-i'+1)\beta}{t}\right) \label{eq:lemma_exact_3}.
 \end{IEEEeqnarray}
 The desired bound now follows by combining \eqref{eq:lemma_exact_1} with \eqref{eq:lemma_exact_3}.

 \section{General File Size bound, functional repair} \label{sec:func}

 In this section, we present the file-size upper bound under functional repair via IFG analysis. Under functional repair, ability to recover a file after a sequence of node failures and repairs is equivalent to multicasting the source file to an arbitrary number of data collectors over the IFG~\cite{dimakis}. The IFG characterizes the data flows from the source to a data collector, and reflects the sequence of failures and repairs in the storage network. The IFG used here (see Fig. \ref{fig:cut}) is a generalization of the one presented in~\cite{genrcTIT} for the case of $t = 1$.

 \begin{figure}
 	\centering
 	\includegraphics[width=85mm]{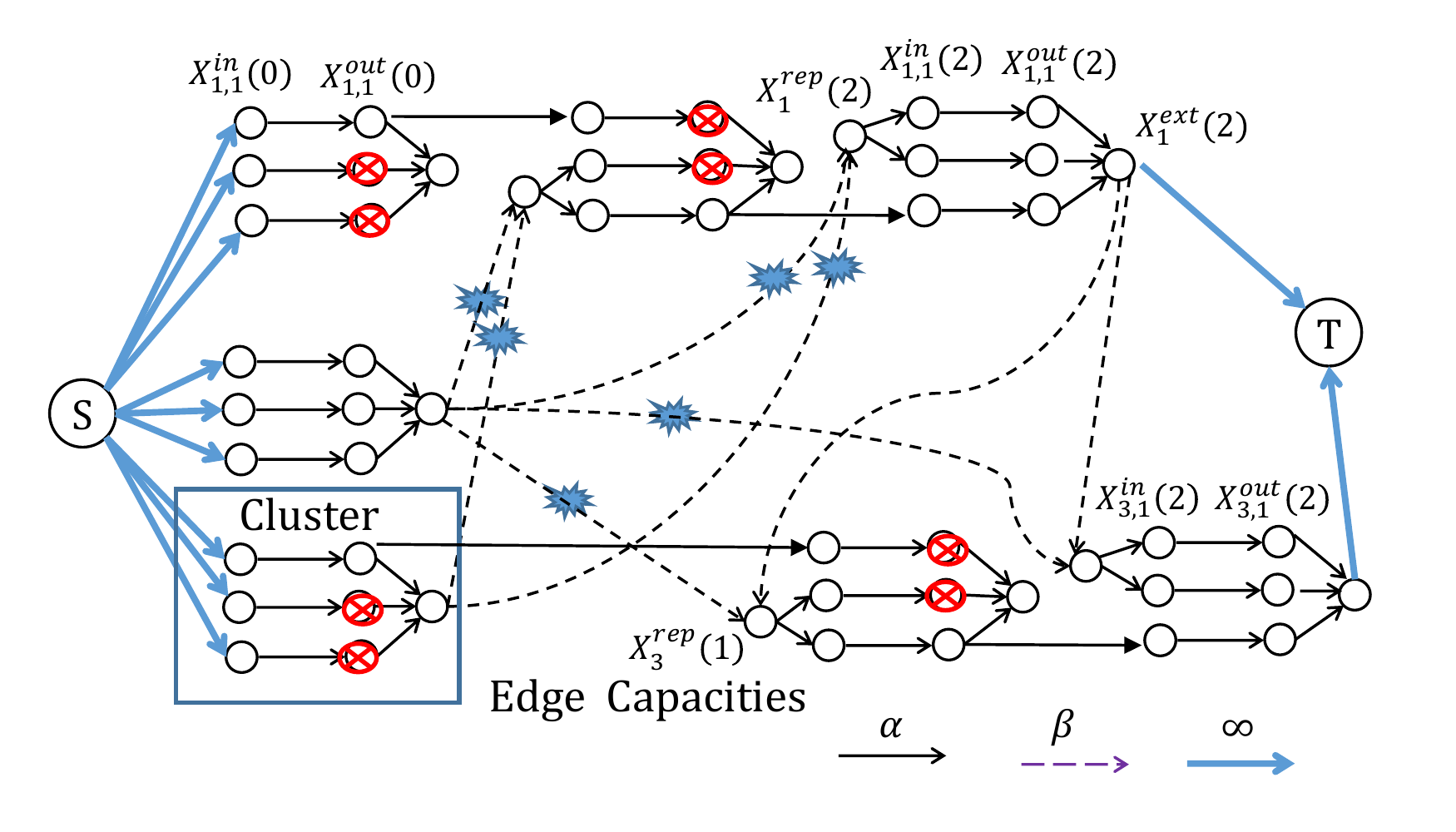}
 	\caption{An illustration of the information flow graph used in cut-set based upper bound for the file-size under functional repair. We assume $(n = 3, k = 2, d= 2) (m = 3, \ell = 0, t = 2)$. Only a subset of nodes are named so as to avoid clutter. Two batches, each of $t = 2$ nodes, fail and get repaired first in cluster $1$ and then in cluster $3$. We also indicate a possible choice of the $S-T$ cut that results in the desired upper bound.  We fail nodes in cluster $3$ instead of cluster $2$ only to make the figure compact. }
 	\label{fig:cut}
 	\vspace{-0.25in}
 \end{figure}
\vspace{-0.1in}
\subsection{Information Flow Graph Model}
Let $\mathcal{X}_i$ denote the physical cluster $i$, and let $X_{i, j}$ denote the physical node $j$ in cluster $i$, $ 1 \leq i \leq n, 1 \leq j \leq m$. In the IFG, $X_{i, j}$  is mapped to the pair of nodes $(X_{i, j}^{in}, X_{i, j}^{out})$ such that the edge $X_{i, j}^{in} \rightarrow X_{i, j}^{out}$ has capacity $\alpha$. The \emph{external} node $X_i^{ext}$ of cluster $i$ serves to transfer data outside the cluster. The $m$ out-nodes connect to $X_i^{ext}$ via edges of capacity $\alpha$.

When a cluster, say $i$, experiences a batch of $t$ failures, the whole cluster becomes inactive and is replaced with a new active cluster. In the new cluster, a  special  \emph{repair} node $X_i^{rep}$ is used to combine local and external helper data, and generate the content of the replacement nodes. The out nodes of the $\ell$ local helper nodes connect to  $X_i^{rep}$ via links of capacity $\alpha$, and the external nodes of the $d$ helper clusters connect to $X_i^{rep}$ via links of capacity $\beta$. Also, $X_i^{rep}$ connects to the in-nodes of the replacement nodes via links of capacity $\alpha$. Further, the $m-t$ nodes, which did not experience failure in the inactive cluster are copied as such in the new active cluster. At any point in time, physical cluster $i$ contains one active and $f_i$ inactive clusters in the IFG where $f_i \geq 0$ denotes the total number of batch failures and repairs in the cluster. We write $\mathcal{X}_i(\tau), 0 \leq \tau \leq f_i$ to denote the cluster in the IFG after the $\tau^{\text{th}}$ (batch) repair associated with cluster $i$, and use $\fset_i(\tau), 0 \leq \tau \leq f_i - 1$ to denote the indices of nodes that fail in $\mathcal{X}_i(\tau)$.  The clusters  $\mathcal{X}_i(0), \ldots, \mathcal{X}_i(f_i-1)$ are inactive, while $\mathcal{X}_i(f_i)$ is active, after $f_i$ repairs.
The nodes of  $\mathcal{X}_i(\tau)$ will be denoted by $X_{i,j}^{in}(\tau), X_{i,j}^{out}(\tau), X_{i}^{ext}(\tau), X_i^{rep}(\tau)$ (there is no repair node if $\tau = 0$).

Finally, the source node $S$ connects to all the $mn$ in-nodes $X_{i,j}^{in}(0), 1 \leq i \leq m, 1 \leq j \leq n$ via links of infinite capacity. The data collector $T$ connects to $k$ external nodes, say $X_{i}^{ext}(f_i), 1 \leq i \leq k$ also via links of infinite capacity.
\vspace{-0.1in}

\subsection{File-Size Upper Bound}

We explain the proof of the bound in \eqref{eq:func_file_size} by considering the special case $(n = 3, k = 2, d = 2)(\alpha, \beta)(m = 3, \ell = 0, t =2)$. A full proof appears in Appendix B. Note that for this special case, $t \nmid (m-\ell)$ and this will help us illustrate the difference between functional and exact repair. Consider the following sequence of $4$ batches of failures and repairs (see Fig. \ref{fig:cut}). Batches $1$ and $2$ are associated with cluster $1$ with $\mathcal{R}_1(0) = \{ 2, 3\}$ and $\mathcal{R}_1(1) = \{ 1, 2\}$.  Batches $3$ and $4$ are associated with cluster $3$ with $\mathcal{R}_3(0) = \{ 2, 3\}$ and $\mathcal{R}_3(1) = \{ 1, 2\}$. There is no local help in this example, cluster $1$ receives external help from $X_{2}^{ext}(0)$ and $X_{3}^{ext}(0)$ for both batches of repairs, while cluster $3$ receives external help from $X_{2}^{ext}(0)$ and $X_{1}^{ext}(2)$ for its repairs. Consider data collection by connecting to $X_1^{ext}(2)$ and $X_3^{ext}(2)$, and consider the $S$-$T$ cut whose edges are found as follows: For disconnecting $X_{1, 1}^{out}(2)$ and $X_{1, 2}^{out}(2)$, we either remove (based on whichever has smaller capacity) the two edges $X_{1, 1}^{in}(2) \rightarrow X_{1, 1}^{out}(2)$ and $X_{1, 2}^{in}(2) \rightarrow X_{1, 2}^{out}(2)$ or the set of helper edges $X_{2}^{ext}(0) \rightarrow X_1^{rep}(2)$ and $X_{3}^{ext}(0) \rightarrow X_1^{rep}(2)$. For disconnecting $X_{1, 3}^{out}(2)$, we either remove the \emph{single} edge $X_{1, 3}^{in}(1) \rightarrow X_{1, 3}^{out}(1)$ or the set of \emph{two} helper edges $X_{2}^{ext}(0) \rightarrow X_1^{rep}(1)$ and $X_{3}^{ext}(0) \rightarrow X_1^{rep}(1)$. The set of edges that disconnects cluster $3$ is similarly found, except that if we choose to disconnect links from external helpers, we only disconnect those from $X_{2}^{ext}(0)$ and not $X_{1}^{ext}(2)$. The value of the cut forms an upper bound for $B$, and is given by $B \leq  \min(2\alpha, d\beta)+ \min(\alpha, d\beta) + \min(2\alpha, (d-1)\beta) + \min(\alpha, (d-1)\beta)$, which is the same as the one give by \eqref{eq:func_file_size}.

\emph{Converse:} We note that it can also be shown that any valid IFG, regardless of the specific sequence of failures and repairs, $B_F^*$ (see \eqref{eq:func_file_size}) is indeed a lower bound on the minimum possible value of any $S$-$T$ cut. Please see Appendix C for a proof of this fact, which establishes the system capacity under functional repair. Note that the RLNC simulation in Fig. \ref{fig:rlnc_filesize} is an experimental verification for the validity of this converse statement.

\bibliographystyle{IEEEtran}
\bibliography{citations}

\begin{thebibliography}{10}
\providecommand{\url}[1]{#1}
\csname url@samestyle\endcsname
\providecommand{\newblock}{\relax}
\providecommand{\bibinfo}[2]{#2}
\providecommand{\BIBentrySTDinterwordspacing}{\spaceskip=0pt\relax}
\providecommand{\BIBentryALTinterwordstretchfactor}{4}
\providecommand{\BIBentryALTinterwordspacing}{\spaceskip=\fontdimen2\font plus
\BIBentryALTinterwordstretchfactor\fontdimen3\font minus
  \fontdimen4\font\relax}
\providecommand{\BIBforeignlanguage}[2]{{%
\expandafter\ifx\csname l@#1\endcsname\relax
\typeout{** WARNING: IEEEtran.bst: No hyphenation pattern has been}%
\typeout{** loaded for the language `#1'. Using the pattern for}%
\typeout{** the default language instead.}%
\else
\language=\csname l@#1\endcsname
\fi
#2}}
\providecommand{\BIBdecl}{\relax}
\BIBdecl

\bibitem{racs}
H.~Abu-Libdeh, L.~Princehouse, and H.~Weatherspoon, ``Racs: a case for cloud
  storage diversity,'' in \emph{Proceedings of the 1st ACM symposium on Cloud
  computing}.\hskip 1em plus 0.5em minus 0.4em\relax ACM, 2010, pp. 229--240.

\bibitem{depsky}
A.~Bessani, M.~Correia, B.~Quaresma, F.~Andr{\'e}, and P.~Sousa, ``Depsky:
  dependable and secure storage in a cloud-of-clouds,'' \emph{ACM Transactions
  on Storage (TOS)}, vol.~9, no.~4, p.~12, 2013.

\bibitem{cyrus}
J.~Y. Chung, C.~Joe-Wong, S.~Ha, J.~W.-K. Hong, and M.~Chiang, ``Cyrus: Towards
  client-defined cloud storage,'' in \emph{Proceedings of the Tenth European
  Conference on Computer Systems}.\hskip 1em plus 0.5em minus 0.4em\relax ACM,
  2015, p.~17.

\bibitem{hitachi}
J.~D. Cook, R.~Primmer, and A.~de~Kwant, ``Compare cost and performance of
  replication and erasure coding,'' \emph{Hitachi Review}, vol.~63, p. 304,
  2014.

\bibitem{dimakis}
A.~G. Dimakis, P.~B. Godfrey, Y.~Wu, M.~J. Wainwright, and K.~Ramchandran,
  ``Network coding for distributed storage systems,'' \emph{Information Theory,
  IEEE Transactions on}, vol.~56, no.~9, pp. 4539--4551, 2010.

\bibitem{genrcTIT}
N.~{Prakash}, V.~{Abdrashitov}, and M.~{M\'edard}, ``{The Storage vs
  Repair-Bandwidth Trade-off for Clustered Storage Systems},'' \emph{ArXiv
  e-prints}, vol. abs/1701.04909, Jan. 2017.

\bibitem{skoglund_partial}
M.~Gerami, M.~Xiao, and M.~Skoglund, ``Two-layer coding in distributed storage
  systems with partial node failure/repair,'' \emph{IEEE Communications
  Letters}, vol.~PP, no.~99, 2017.

\bibitem{ford2010availability}
D.~Ford, F.~Labelle, F.~I. Popovici, M.~Stokely, V.-A. Truong, L.~Barroso,
  C.~Grimes, and S.~Quinlan, ``Availability in globally distributed storage
  systems.'' in \emph{OSDI}, vol.~10, 2010, pp. 1--7.

\bibitem{KoetterMedard}
R.~Koetter and M.~M\'edard, ``An algebraic approach to network coding,''
  \emph{IEEE/ACM Transactions on Networking}, vol.~11, no.~5, pp. 782--795, Oct
  2003.

\bibitem{rlnc}
T.~Ho, M.~M\'edard, R.~Koetter, D.~R. Karger, M.~Effros, J.~Shi, and B.~Leong,
  ``A random linear network coding approach to multicast,'' \emph{IEEE
  Transactions on Information Theory}, vol.~52, no.~10, pp. 4413--4430, Oct
  2006.

\bibitem{shum_coop}
K.~W. Shum and Y.~Hu, ``Cooperative regenerating codes,'' \emph{IEEE
  Transactions on Information Theory}, vol.~59, no.~11, 2013.

\bibitem{kermarrec_coop}
A.-M. Kermarrec, N.~Le~Scouarnec, and G.~Straub, ``Repairing multiple failures
  with coordinated and adaptive regenerating codes,'' in \emph{Network Coding
  (NetCod), 2011 International Symposium on}.\hskip 1em plus 0.5em minus
  0.4em\relax IEEE, 2011.

\bibitem{cadambe_asymptotic}
V.~R. Cadambe, S.~A. Jafar, H.~Maleki, K.~Ramchandran, and C.~Suh, ``Asymptotic
  interference alignment for optimal repair of mds codes in distributed
  storage,'' \emph{IEEE Transactions on Information Theory}, vol.~59, no.~5,
  pp. 2974--2987, 2013.

\bibitem{ankit_centralized}
A.~S. Rawat, O.~O. Koyluoglu, and S.~Vishwanath, ``Centralized repair of
  multiple node failures with applications to communication efficient secret
  sharing,'' \emph{CoRR}, vol. abs/1603.04822, 2016.

\bibitem{plee_isit2016_doubleregen}
Y.~Hu, P.~P.-C. Lee, and X.~Zhang, ``Double regenerating codes for hierarchical
  data centers,'' in \emph{Information Theory (ISIT), 2015 IEEE International
  Symposium on}.\hskip 1em plus 0.5em minus 0.4em\relax IEEE, 2016.

\bibitem{clust_stor_Moon}
J.~Sohn, B.~Choi, S.~W. Yoon, and J.~Moon, ``Capacity of clustered distributed
  storage,'' \emph{CoRR}, vol. abs/1610.04498, 2016.

\bibitem{gaston2013realistic}
B.~Gast{\'o}n, J.~Pujol, and M.~Villanueva, ``A realistic distributed storage
  system: the rack model,'' \emph{arXiv preprint arXiv:1302.5657}, 2013.

\bibitem{Gaston_nonhom}
J.~Pernas, C.~Yuen, B.~Gast\'on, and J.~Pujol, ``Non-homogeneous two-rack model
  for distributed storage systems,'' in \emph{Information Theory Proceedings
  (ISIT), 2013 IEEE International Symposium on}, July 2013.

\bibitem{ozan_xyregen}
G.~Calis and O.~O. Koyluoglu, ``Architecture-aware coding for distributed
  storage: Repairable block failure resilient codes,'' \emph{CoRR}, vol.
  abs/1605.04989, 2016.

\bibitem{rbt}
N.~B. Shah, K.~V. Rashmi, P.~V. Kumar, and K.~Ramchandran, ``Distributed
  storage codes with repair-by-transfer and nonachievability of interior points
  on the storage-bandwidth tradeoff,'' \emph{IEEE Transactions on Information
  Theory}, vol.~58, no.~3, pp. 1837--1852, March 2012.

\end{thebibliography}

\newpage

\appendices

\section{Proof of Lemma \ref{lem:exact}} \label{app:lem_exact}

We will  show that the permutation $\sigma_{i', S_i}$ is such that
\begin{equation} \label{eq:lem_exact1_rep}
H\left(Y_{i', \sigma_{i', S_i}(j')}|\mathcal{Y}(S_i), \widetilde{\mathcal{Y}}(i', S_i, j') \right)   \leq  \min\left(\alpha, \frac{(d-i)\beta}{t}\right),
\end{equation}
for all $j' \in \{m - b + 1, m - b + 2, \ldots, m\}$, where
\begin{equation} \label{eq:lem_exact2_rep}
\widetilde{\mathcal{Y}}(i', S_i,j')   = Y_{i', [1 , \ell]} \cup  \{Y_{i', \sigma_{i', S_i}(j)}, j \in [\ell+1 , j'-1 ] \}.
\end{equation}
Consider the variable $j'$ appearing in \eqref{eq:lem_exact1_rep}, and let $j' = m - x$ for some $x, 0 \leq x \leq b-1$ so that using \eqref{eq:perm} we have, $\sigma_{i', S_i}(j') = j_{m - x}$. Consider the definition of $(j_{m - x}, \mathcal{V}_{m - x})$ in \eqref{eq:def_j}; we then know that
\begin{eqnarray} \label{eq:lemma_proof_temp1}
H\left( Y_{i', j_{m - x}}| \mathcal{V}_{m - x}, \mathcal{Y}(S_i), Y_{i', [1:\ell]} \right) & \leq & \nonumber \\
& & \hspace{-1in} H\left( Y_{i', j_{p}}| \mathcal{V}, \mathcal{Y}(S_i), Y_{i', [1:\ell]} \right),
\end{eqnarray}
for all $\mathcal{V} \subset \{Y_{i', j_{\ell + 1}}, Y_{i', j_{\ell + 2}}, \ldots, Y_{i',j_{m-x}}\} \backslash \{ Y_{i', j_p}\}$ such that $|\mathcal{V}| = t - 1$, and for all $p, \ell + 1 \leq p \leq m - x -1$. Towards proving \eqref{eq:lem_exact1_rep}, first of all, observe that
\begin{eqnarray} \label{eq:lem_exact1_step1}
H\left(Y_{i', \sigma_{i', S_i}(j')}|\mathcal{Y}(S_i), \widetilde{\mathcal{Y}}(i', S_i, j') \right)  & \leq &  \nonumber \\
& & \hspace{-1.7in} H\left(Y_{i', \sigma_{i', S_i}(j')}|\mathcal{Y}(S_i), \mathcal{V}_{m - x}, Y_{i', [1:\ell]} \right).
\end{eqnarray}
This follows from \eqref{eq:lem_exact2_rep} and because of the fact that $\mathcal{V}_{m - x} \subset \{Y_{i', j_{\ell + 1}}, Y_{i', j_{\ell + 2}}, \ldots, Y_{i',j_{m - x-1}}\}$. Without loss of generality, assume that
$\mathcal{V}_{m - x} = \{Y_{i', j_{\ell + 1}}, Y_{i', j_{\ell + 2}}, \ldots, Y_{i',j_{\ell + t-1}}\}$.
Next, from the exact repair condition given in \eqref{eq:exact_rep}, we know that
\begin{eqnarray}
\min(t\alpha, (d-i)\beta) & \geq & H\left(Y_{i', \sigma_{i', S_i}(j')}, \mathcal{V}_{m - x} | \mathcal{Y}(S_i), Y_{i', [1:\ell]} \right) \nonumber \\
& \hspace{-1.8in} = & \hspace{-1in} \sum_{p = \ell + 1}^{\ell + t-1} H\left(Y_{i',j_p} | Y_{i',j_{\ell + 1}}, \ldots, Y_{i',j_{p-1}}, \mathcal{Y}(S_i), Y_{i', [1:\ell]} \right) \ + \nonumber \\
&& \hspace{-0.5in}  H\left(Y_{i', \sigma_{i', S_i}(j')} |  \mathcal{V}_{m - x}, \mathcal{Y}(S_i), Y_{i', [1:\ell]} \right) \nonumber \\
& \hspace{-1.8in} \geq & \hspace{-1in}\sum_{p = \ell+1}^{\ell+t-1} H\left(Y_{i',j_p} | \mathcal{V}_{j_p}, \mathcal{Y}(S_i),Y_{i', [1:\ell]}  \right) \ + \nonumber \\
&& \hspace{-0.5in} H\left(Y_{i', \sigma_{i', S_i}(j')} |  \mathcal{V}_{m - x}, \mathcal{Y}(S_i), Y_{i', [1:\ell]}\right),  \label{eq:lemma_proof_temp2}
\end{eqnarray}
where $\mathcal{V}_{j_p} = \mathcal{V}_{m - x} \backslash \{ Y_{i', j_p}\} \cup \{ Y_{i', \sigma_{i', S_i}(j')}\}$. Noting that $|\mathcal{V}_{j_p}| = t-1$, we see that each term under the first summation in \eqref{eq:lemma_proof_temp2} can be lower bounded using \eqref{eq:lemma_proof_temp1}, i.e.,
\begin{eqnarray}
H\left(Y_{i',j_p} | \mathcal{V}_{j_p}, \mathcal{Y}(S_i), Y_{i', [1:\ell]} \right) & & \nonumber \\
& \hspace{-2in} \geq & \hspace{-1in} H\left( Y_{i', j_{m - x}}| \mathcal{V}_{m - x}, \mathcal{Y}(S_i), Y_{i', [1:\ell]} \right) \nonumber \\
& \hspace{-2in} = & \hspace{-1in} H\left(Y_{i', \sigma_{i', S_i}(j')} |  \mathcal{V}_{m - x}, \mathcal{Y}(S_i), Y_{i', [1:\ell]} \right). \label{eq:lemma_proof_temp3}
\end{eqnarray}
Combining \eqref{eq:lemma_proof_temp3} with \eqref{eq:lemma_proof_temp2}, it follows that
\begin{equation}
H\left(Y_{i', \sigma_{i', S_i}(j')} |  \mathcal{V}_{m - x}, \mathcal{Y}(S_i), Y_{i', [1:\ell]} \right)  \leq  \min\left(\alpha, \frac{(d-i)\beta}{t} \right). \label{eq:lemma_proof_temp4}
\end{equation}
The proof of lemma now follows by combining \eqref{eq:lemma_proof_temp4} with \eqref{eq:lem_exact1_step1}.

\section{Proof of upper-bound \ref{eq:func_file_size}} \label{app:func_upper}
We prove that under functional repair the file-size is upper-bounded by

\begin{IEEEeqnarray*}{rCl}
	B &\leq & B_F^*  =  \ell k\alpha + a \sum_{i=1}^{k} \min (t\alpha, (d-i+1)\beta ) \\
	& &+ \sum_{i=1}^{k} \min (b\alpha, (d-i+1)\beta ). \label{eq:func_file_size_proof}
\end{IEEEeqnarray*}

Let $[A,B]=\{\mbox{integer }x: A \leq x\leq B\}$, and $[B] = [1,B]$.

To show the bound, it is enough to demonstrate a sequence of batch failures and a set of $k$ clusters used by a data-collector, such that there exists a cut between the source and the data-collector with capacity no more than $B_F^*$. In the example sequence, that we consider, the clusters $1$ to $k$ are used for data-collection and experience node failures. At each of these clusters $a+1$ batch failures occur. They jointly cover the first $m-\ell$ nodes of a cluster. Specifically, at cluster $i\in [k]$, the first batch failure affects the last $t$ of these nodes: $\mathcal{R}_i(0)=\{ m-\ell-t+1,\dotsc, m-\ell\}$. The remaining batch failures affect disjoint sets of $t$ nodes starting from the first node $X_{i,1}$: $\mathcal{R}_i(1)=\{1,\dotsc, t\}$, $\mathcal{R}_i(2)=\{t+1,\dotsc, 2t\}$, until $\mathcal{R}_i(a)=\{(a - 1) t,\dotsc, a t\}$.

In all cases, the last $\ell$ nodes in a cluster provide the local helper data. For repairs in cluster $i$, clusters $1, \dotsc, i-1$ and $n-(d-i), \dotsc, n$ serve as helper clusters.
Failures first occur in cluster $1$, then in clusters $2,3$, etc. until cluster $k$.

In the IFG, corresponding to the described failure sequence, cluster $\mathcal{X}_i(a+1)$ is active for each $i\in k$. Let $\tau_j$ be the such that the cluster $\mathcal{X}_i(\tau_j)$ appears in the IFG right after the last repair of node $X_{i,j}$ (we say ``last repair" since nodes whose indices belong to  $\mathcal{R}_i(0) \cap \mathcal{R}_i(a)$ fail twice in our sequence of failures; other nodes in cluster $i$ fail only once). Consider a cut-set $(IFG_S, IFG_T)$ consisting of the following edges:
\begin{itemize}
	\item $X_{i,j}^{in}(a+1) \overset{\alpha}{\to} X_{i,j}^{out}(a+1), \forall i\in [k], j\in [m-\ell+1, m]$. Total capacity of these edges is $\ell k \alpha$.
	\item For all $i\in [k]$:
	\begin{itemize}
		\item Edge set $X_{i}^{rep}(\tau_j)\overset{\alpha}{\to} X_{i,j}^{in}(\tau_j), j\in [at]$, or edge set $X_{i'}^{ext}(0) \overset{\beta}{\to} X_{i}^{rep}(\tau_j) \forall i'\in [n-(d-i), n], j\in \{t, 2t, \dotsc, at \}$, whichever set capacity is smaller. Total capacity of these edges is $a\min (t\alpha, (d-i+1)\beta)$.
		\item If $b>0$: edge set $X_{i}^{rep}(\tau_j)\overset{\alpha}{\to} X_{i,j}^{in}(\tau_j), j\in [at+1, m-\ell]$,
		or edge set $X_{i'}^{ext}(0) \overset{\beta}{\to} X_{i}^{rep}(\tau_j) \forall i'\in [n-(d-i), n], j=m-\ell$, whichever set capacity is smaller. Total capacity of these edges is $\min (b\alpha, (d-i+1)\beta)$.
	\end{itemize}
\end{itemize}
The value of the cut is given by $\ell k\alpha + a \sum_{i=1}^{k} \min (t\alpha, (d-i+1)\beta ) +  \sum_{i=1}^{k} \min (b\alpha, (d-i+1)\beta )=B_F^*$.
\qedsymbol

\section{Min-Cut for IFG, Optimality of $B_F^*$ for general $\ell$} \label{app:func_converse}

We now show that for any valid IFG, regardless of the specific sequence of failures and repairs, $B_F^*$ is indeed a lower bound on the minimum possible value of any $S$-$T$ cut. Consider a cut of IFG, and let IFG$_S$ and IFG$_T$ be the two disjoint parts associated with nodes $S$ and $T$, respectively. Without loss of generality, we only consider cuts such that IFG$_T$ contains at least $k$ external nodes corresponding to active clusters. Consider a topological sorting of the IFG nodes such that: $1)$ an edge exists between two nodes $A$ and $B$ only if $A$ appears before $B$ in the sorting, and $2)$ all in-, out-, external, and repair nodes (if $\tau > 0$) of the cluster $\mathcal{X}_i(\tau)$ appear together in the sorted order, $\forall i,\tau$.

Consider the sequence $\mathcal{E}$ of all the external nodes in both active and inactive clusters in IFG$_T$ in their sorted order. Let $Y_1$ denote the first node in $\mathcal{E}$. Without loss of generality let $Y_1 = X_1^{ext}(\tau_1)$, for some $\tau_1$. In this case, consider the subsequence of $\mathcal{E}$ which is obtained after excluding all the external nodes associated with $\mathcal{X}_1$ from $\mathcal{E}$. Let $Y_2$ denote the first  external node in this subsequence. We continue in this manner until we find the first $k$ external nodes $\{ Y_1, Y_2, \ldots, Y_k\}$ in $\mathcal{E}$, such that each of the $k$ nodes corresponds to a distinct physical cluster. Without loss of generality, let us also assume that $Y_i = X_i^{ext}(\tau_i), 2 \leq i \leq k$, for some $\tau_i$. If $\tau_i = 0$,  then clearly cluster $i$ contributes (at least) $m\alpha$ to the cut. Thus let us assume that $\tau_i > 0, 1 \leq i \leq k$.

Consider the $m$ out-nodes $X_{i, 1}^{out}(\tau_i), \ldots, X_{i, m}^{out}(\tau_i)$ that connect to $X_i^{ext}(\tau_i)$.
For each $j\in [1,m]$, either $X_{i,j}^{out}(\tau_i)$ is in IFG$_S$ or there exists a minimal $\tau_{i, j} \in [0,\tau_i]$ such that $X_{i,j}^{out}(\tau_{i,j})\in \text{IFG}_T$. Consider those values of $j\in [1,m]$ for which all the following conditions hold:
\begin{gather}
X_{i,j}^{out}(\tau_i), X_{i,j}^{in}(\tau_{i,j}) \in \text{IFG}_T, j\in \fset_i(\tau_{i,j}-1), \nonumber \\
X_{i}^{rep}(\tau_{i,j}) \in \text{IFG}_T.  \label{eq:mi_conditions}
\end{gather}

Let there be $m_i\in [0,m]$ of such values, and, without loss of generality, let them be $m-m_i+1, \dotsc, m$. Also without loss of generality, let indices $j$ be sorted in the order of increasing $\tau_{i,j}$, i.e. $j_1<j_2$ implies $\tau_{i,j_1}\leq \tau_{i,j_2}$. For each $j \in [m-m_i+1, m]$, $\Sigma_{i,j} \triangleq \{ j': \tau_{i,j'} = \tau_{i,j}, j'\in [m-m_i+1, m] \}$ is a contiguous set of at most $t$ indices of the nodes with the same $\tau_{i,j}$, and which are repaired together from the same repair node.
Let $\mathcal{S}_i = \{\mbox{distinct } (\min \Sigma_{i,j}-1), \forall j \in [m-m_i+1, m]\} \subseteq [m-m_i, m-1]$ be the set of indices of the nodes preceding all contiguous groups $\Sigma_{i,j}$. Note that by $\min \Sigma_{i,j}$ we mean the minimum element contained in the set $\Sigma_{i,j}$. The set $\mathcal{S}_i$ is in one-to-one correspondence with the set of the repair nodes in \eqref{eq:mi_conditions} for $j \in [m-m_i+1, m]$. Note that $m-m_i$ is always an element of $\mathcal{S}_i$.

In order to relay helper data to $X_{i,j}^{in}(\tau_{i,j})$ for all $j \in [m-m_i+1, m]$, the number of these repair nodes should be at least $\lceil m_i/t \rceil$, and $|\mathcal{S}_i| \geq \lceil m_i/t \rceil$.
Each of these repair nodes connects to $d$ external nodes in other clusters. By construction of $\mathcal{E}$, at most $i-1$ of those external nodes can be in IFG$_T$.  Thus, each repair node contributes at least $(d-i+1)\beta$ of external helper data to the cut value. In addition, each repair node $X_i^{rep}(\tau_{i,j})$ connects to $\ell$ local nodes. By \eqref{eq:mi_conditions} and by construction of $\mathcal{S}_i$ and sorting of $\tau_{i,j}$, only nodes with indices $\{1, 2, \dotsc, j'\}$ out of these $\ell$ can be in IFG$_T$, where $j' = \min \Sigma_{i,j}-1$ is the corresponding element of $\mathcal{S}_i$.  Thus, repair node $X_i^{rep}(\tau_{i,j})$ contributes at least $(\ell-j')^+\alpha$ of local helper data to the cut value.

The contribution to the cut value of those $m-m_i$ indices of $j\in[1,m-m_i]$, which do not satisfy \eqref{eq:mi_conditions}, is at least $\alpha$ each.

Based on the observations above, the overall cut value is lower-bounded by
{ \begin{IEEEeqnarray}{rCl}
		\IEEEeqnarraymulticol{3}{l}{\text{mincut}(S-T)} \nonumber \\
		& \geq & \sum_{i=1}^{k} \big((m-m_i)\alpha  + \big\lceil \frac{m_i}{t} \big\rceil (d-i+1)\beta + \sum_{j'\in \mathcal{S}_i} (\ell-j')^+ \alpha \big). \IEEEeqnarraynumspace \label{eq:funcfilesizeLB1}
	\end{IEEEeqnarray}}

	Consider a particular value of $i\in [1,k]$ and the corresponding summation term in \eqref{eq:funcfilesizeLB1}. Let us assume that $m-m_i\geq \ell$, and $m_i = a_i t + b_i \leq m- \ell, b_i \in [0,t-1]$. Then the third term in \eqref{eq:funcfilesizeLB1} is zero, and
	{ \begin{IEEEeqnarray*}{rCl}
			\IEEEeqnarraymulticol{3}{l}{(m-m_i)\alpha  + \lceil m_i/t \rceil (d-i+1)\beta} \\
			& = & m \alpha - (a_i t + b_i) \alpha + (a_i + 1_{b_i>0})(d-i+1)\beta \\
			& = & m \alpha - a_i(t \alpha -(d-i+1)\beta) - b_i \alpha + 1_{b_i>0}(d-i+1)\beta \IEEEeqnarraynumspace \\
			& \overset{(1)}{\geq} & \ell \alpha + (m-\ell) \alpha - a(t \alpha -(d-i+1)\beta)^+ \\
			\IEEEeqnarraymulticol{3}{l}{\quad -(b \alpha - (d-i+1)\beta)^+} \\
			& = & \ell \alpha + a(t\alpha - (t \alpha -(d-i+1)\beta)^+) \\
			\IEEEeqnarraymulticol{3}{l}{\quad +(b\alpha - (b \alpha - (d-i+1)\beta)^+)} \\
			& = & \ell \alpha + a\min(t \alpha,(d-i+1)\beta) + \min(b\alpha,(d-i+1)\beta) \IEEEeqnarraynumspace \\
			& \triangleq & C_i,
		\end{IEEEeqnarray*}}
		where $(1)$ follows, because $a_i t + b_i = m_i \leq m - \ell = a t + b$, $a_i \leq a$, and, if $a_i=a$, $b_i\leq b$.

		On the other hand, if $m-m_i = \ell - \mu_i < \ell$, and $m_i > m-\ell = at+b, \ell - (m-m_i) = \mu_i >0$, then we have
		{ \begin{IEEEeqnarray*}{rCl}
				\IEEEeqnarraymulticol{3}{l}{(m-m_i)\alpha  + \lceil m_i/t \rceil (d-i+1)\beta + \sum_{j'\in \mathcal{S}_i} (\ell-j')^+ \alpha} \\
				& \geq & (\ell-\mu_i) \alpha + (a + 1_{b>0})(d-i+1)\beta + (\ell-(m-m_i))\alpha \IEEEeqnarraynumspace \\
				\IEEEeqnarraymulticol{3}{l}{\quad + \sum_{\substack{j'\in \mathcal{S}_i\\ j'>m-m_i}} (\ell-j')^+ \alpha} \\
				& = & \ell \alpha + (a + 1_{b>0})(d-i+1)\beta + \sum_{\substack{j'\in \mathcal{S}_i\\ j'>m-m_i}} (\ell-j')^+ \alpha \\
				& \geq & C_i,
			\end{IEEEeqnarray*}}
			where $C_i$ is the lower-bound for the case $m-m_i\geq \ell$.

			Since $B_F^* = \sum_i C_i$, it is indeed a lower-bound on the file-size. This proves tightness of $\ref{eq:func_file_size}$.

\end{document}